\documentclass[12pt]{article}

\newcommand{\blind}{0}

\usepackage{amsmath}
\usepackage{graphicx}
\usepackage{enumerate}
\usepackage{natbib}
\usepackage{url} 

\usepackage{mathabx}
\usepackage{bbm}
\renewcommand*{\P}{\mathbbm{P}}
\newcommand*{\E}{\mathbbm{E}}
\newcommand*{\Var}{\mathrm{Var}}

\renewcommand*{\vec}[1]{\boldsymbol{#1}}
\newcommand*{\mat}[1]{\mathrm{#1}}
\newcommand*{\set}[1]{\mathrm{#1}}

\newcommand*{\Ind}{\mathbbm{1}}

\newcommand*{\giv}{\;|\;}
\newcommand*{\mgiv}{\;\middle|\;}
\newcommand{\ind}{\perp\!\!\!\perp} 

\usepackage{amsthm}
\theoremstyle{plain}
\newtheorem{lem}{Lemma}
\newtheorem{thm}{Theorem}

\theoremstyle{definition}
\newtheorem{asm}{Assumption}
\newtheorem{tcnd}[asm]{Simplifying Assumption}

\newtheorem{exm}{Example}

\usepackage{mathtools}
\usepackage{color}
\usepackage{booktabs}

\begin{document}

\def\spacingset#1{\renewcommand{\baselinestretch}%
{#1}\small\normalsize} \spacingset{1}


\if0\blind
{
  \title{\bf Overcoming Repeated Testing Schedule Bias in Estimates of Disease Prevalence}
  \author{
    Patrick M. Schnell \\
    Division of Biostatistics, College of Public Health, The Ohio State University \\
    and \\
    Matthew Wascher \\
    Department of Mathematics, University of Dayton \\
    and \\
    Grzegorz A. Rempala \\
    Division of Biostatistics, College of Public Health, The Ohio State University
  }
  \maketitle
} \fi

\if1\blind
{
  \bigskip
  \bigskip
  \bigskip
  \begin{center}
    {\LARGE\bf Overcoming Repeated Testing Schedule Bias in Estimates of Disease Prevalence}
\end{center}
  \medskip
} \fi

\bigskip
\begin{abstract} 
  During the COVID-19 pandemic, many institutions such as universities and workplaces implemented testing regimens with every member of some population tested longitudinally, and those testing positive isolated for some time.
  Although the primary purpose of such regimens was to suppress disease spread by identifying and isolating infectious individuals, testing results were often also used to obtain prevalence and incidence estimates.
  Such estimates are helpful in risk assessment and institutional planning and various estimation procedures have been implemented, ranging from simple test-positive rates to complex dynamical modeling.
  Unfortunately, the popular test-positive rate is a biased estimator of prevalence under many seemingly innocuous longitudinal testing regimens with isolation.
  We illustrate how such bias arises and identify conditions under which the test-positive rate is unbiased.
  Further, we identify weaker conditions under which prevalence is identifiable and propose a new estimator of prevalence under longitudinal testing.
  We evaluate the proposed estimation procedure via simulation study and illustrate its use on a dataset derived by anonymizing testing data from \if1\blind{[BLINDED: Large United States University]}\else{The Ohio State University}\fi.
\end{abstract}

\noindent%
{\it Keywords:} COVID-19, inverse probability weighting, longitudinal census
\vfill

\newpage
\spacingset{1.9} 
\section{Introduction}

In the wake of the first signs of a global COVID-19 pandemic in early 2020, various tests for detecting SARS-CoV-2 were developed across the world within days of the release of the virus genome \citep{mina2021covid, corman2022detection}, with some countries who early on invested in large-scale testing capacity being able to control the SARS-CoV-2 transmission \citep{baker2022successful}.
The various testing strategies were also often essential parts of the gradual lifting of lockdowns and the relaxing of mask-wearing rules \citep{school21, reopening20}.
Accordingly, during different periods of the COVID-19 pandemic, many institutions implemented comprehensive testing regimens in which every member was tested longitudinally.
Notable examples include several universities \citep{school21, college21, chang2021repeat}, workplaces \citep{work22}, and sports leagues \citep{mba21}.
The primary purpose of these regimens was to suppress disease spread within the population by identifying and isolating infectious individuals and potentially quarantining their close contacts.
Thus the regimens were designed to attempt to detect as early as possible all infectious individuals during their infectious period, for example by requiring each individual to be tested at least or exactly once during each calendar week, or requiring that an individual go no more than a set small number of days between tests \citep{frazier2022modeling, chang2021repeat}.

A secondary goal of comprehensive longitudinal testing regimens was to provide frequent estimates of prevalence and incidence within the population.
Such estimates may be helpful in risk assessment (e.g., how safe is it to hold gatherings?) and other institutional planning (e.g., isolation and quarantine capacity).
Various methods of estimating prevalence have been applied.
One particularly simple and popular method is based on the test-positive rate (TPR): the number of positive tests divided by the total number of tests administered.
The TPR on a given day has been interpreted as an estimate of prevalence on that day \citep{kahanec2021impact}.
The intuition behind the use of the TPR as an estimator of prevalence arises from sampling arguments: if a sample representative of the population is tested, then the proportion of infectious individuals within that sample (the TPR under perfect test sensitivity and specificity) is an unbiased estimator of the proportion of infectious individuals within the population (the prevalence).
\cite{nicholson2022improving} provide a framework for debiasing local prevalence estimates from a non-representative surveillance population by incorporating information from broader representative samples.
An alternative and more complicated approach is via direct modeling of the process of disease spread, as by \cite{quick2021regression}, which additionally combines confirmed case reporting with seroprevalence data to handle under-ascertainment and unreliable reporting.

Unfortunately, the test-positive rate described above is a biased estimator of prevalence under large classes of natural longitudinal testing regimens when those testing positive are subsequently isolated.
This bias generally arises due to associations between the probability of testing on a given day and the time since the last test.
For instance, if individuals not in isolation or quarantine are tested exactly once per calendar week, the individuals eligible to be tested on a given day late in the week (because they have not yet been tested that week) are more likely to be infectious than those ineligible for testing (because they have already been tested that week and are known to have not been infectious at the time), all else equal.
\cite{hay2021intuition} argue that the TPR in a repeatedly tested population falls between incidence and prevalence in the long run, but do not consider bias due to within-week scheduling.
Estimation methods that involve modeling the process of disease spread may avoid such biases, but generally involve mechanistic assumptions and may require specific expertise and extensive computational resources to implement.

In this work we give detailed illustrations of how the bias of the TPR as an estimator of prevalence arises and describe a necessary and sufficient condition under which the TPR is unbiased under some simplifying assumptions (Section~\ref{sec:bias}).
Further, we identify a set of weaker conditions under which the TPR may be biased but prevalence may be estimated without bias via a Horvitz-Thompson--type estimator (\cite{horvitz1952generalization}, Section~\ref{sec:identifiability}), including under known imperfect test sensitivity and specificity.
We evaluate the Horvitz-Thompson (HT) estimator via simulation study (Section~\ref{sec:simulation}) and illustrate its application to a dataset derived by anonymizing comprehensive longitudinal testing data of a student population at \if1\blind{[BLINDED: Large United States University]}\else{The Ohio State University}\fi\ (Section~\ref{sec:example}).
We conclude with a discussion of strengths, weaknesses, and potential for future work (Section~\ref{sec:discussion}).


\section{Disease and testing process}
\label{sec:processes}

We employ two parallel formulations of the joint disease and testing process: one set-theoretic and one time-to-event.
The set-theoretic formulation is generally parsimonious when defining and manipulating estimators based on sampling at a specific time with minimal intrusion by time-evolution considerations, while the time-to-event formulation is convenient for describing phenomena arising from time-evolution, e.g., associations between probability of infectiousness and probability of testing.
The two formulations are mathematically equivalent, and both allow us to consider very general circumstances without positing a specific mechanistic model.

\subsection{Set-theoretic formulation}

Consider a discrete-time compartmental model in which at time $t > 0$ each of $N$ individuals are in one of three compartments: sets denoted by upright symbols $\set{W}(t)$ (Well), $\set{I}(t)$ (Infectious, the relative size of which we wish to estimate), or $\set{R}(t)$ (Removed).
Although we use language similar to that of compartmental models for infectious diseases, the considerations here apply to non-infectious conditions as well.
Note also differences from common compartmental models for infectious diseases, e.g., the SIRS model: individuals in $\set{W}(t)$ are not necessarily susceptible (they may be immune for a time following a previous recovery), and individuals in $\set{R}$ are not necessarily dead or recovered (they are simply removed from the surveilled population in some way).
For individuals indexed by $i$, let vectors $\vec{W}(t)$, $\vec{I}(t)$, and $\vec{R}(t)$ be the vectors indicating membership in $\set{W}(t)$, $\set{I}(t)$, and $\set{R}(t)$, respectively, so that, e.g., $W_i(t) = 1$ if individual $i$ is in $\set{W}(t)$ and $0$ otherwise.
The individual subscript $i$ may be dropped where the result is unambiguous, and we will denote sums over all $i$ by a $+$ subscript, e.g., $W_+(t) = \sum_{i=1}^{N} W_i(t)$.
Other sets will be denoted similarly.

Suppose that the $N$ total individuals in the population participate in a disease surveillance scheme in which each member is tested repeatedly for the condition.
We do not posit a particular mechanism for transitions from $\set{W}$ to $\set{I}$ or the reverse (natural recoveries go from $\set{I}$ to $\set{W}$), and we assume transitions from $\set{W}$ or $\set{I}$ to $\set{R}$ occur immediately after a positive result is obtained, and only then (no deaths without testing, and self-isolating individuals remain in $\set{I}$ or $\set{W}$).
At some point after entry into $\set{R}$, an individual may be cleared to re-enter $\set{W}$, and this clearance is fully observed.
Thus we have the discrete process:
\begin{enumerate}
\item
  The compartments $(\set{W}(t), \set{I}(t), \set{R}(t))$ represent the system state at the start of day $t$.
\item \label{it:tested-subset}
  A subset $\set{D}(t)$ of the \emph{non-removed} population $\set{W}(t) \cup \set{I}(t)$ is tested on day $t$, with some subset of tests returning positive results, $\set{Y}(t) \subseteq \set{D}(t)$.
  Under perfect sensitivity and specificity, $\set{Y}(t) = \set{D}(t) \cap \set{I}(t)$.
  These individuals will be moved to $\set{R}$ the following day, regardless of other events.
\item
  Some subset of well individuals $\set{Q}(t) \subseteq \set{W}(t)$ receive infecting exposure during the day, but will not be infectious (or detectable via testing) until the following day.
\item
  Some subset of infectious individuals $\set{U}(t) \subseteq \set{I}(t)$ will recover undetected.
\item 
  Some subset of removed individuals $\set{S}(t) \subseteq \set{R}(t)$ are cleared to return to the non-removed population the following day.
\item
  Compartments are updated as
  \begin{equation}
    \begin{aligned}
      \set{W}(t+1) &= \left[ \set{W}(t) \cup \set{U}(t) \cup \set{S}(t) \right] \setminus \left[ \set{Q}(t) \cup \set{Y}(t) \right], \\
      \set{I}(t+1) &= \left[ \set{I}(t) \cup \set{Q}(t) \right] \setminus \left[ \set{U}(t) \cup \set{Y}(t) \right], \\
      \set{R}(t+1) &= \left[\set{R}(t) \cup \set{Y}(t) \right] \setminus \set{S}(t).
    \end{aligned}
  \end{equation}
\end{enumerate}
Our central point, illustrated in Section~\ref{sec:bias}, is that, depending on the specific subset of the population tested in step \eqref{it:tested-subset} above, the TPR for a day $Y_+(t) / D_+(t)$ may not be an unbiased estimator of the prevalence of infection in either the entire population ($I_+(t) / N$) or the non-removed population ($I_+(t) / \left[ N - R_+(t) \right]$), even with perfect test sensitivity and specificity, under some seemingly innocuous mechanisms for selecting $\set{D}(t)$.

We do not allow for the ``exposed'' compartment used in many infectious disease models, representing some delay between exposure and either infectiousness or detectability.
In the former case, the distinction is not relevant for our approach to prevalence estimation as long as we are only concerned with estimating the prevalence of infectiousness: we may simply redefine $\set{Q}(t)$ as not those who are exposed, but those who are about to become infectious.
Delayed detectability is a matter of time-varying sensitivity, discussed in Section~\ref{sec:simulation-violate}.

\subsection{Time-to-event formulation}

The set-theoretic formulation of $(\set{W}(t), \set{I}(t), \set{R}(t))$ may be helpfully re-expressed as a deterministic function of a time-to-event process for infectious exposure and the testing process.
Let $C_{il}$ be the $l$th time individual $i$ is cleared to enter the monitored population at the next timepoint (in $\set{R}(C_{il})$ then $\set{W}(C_{il}+1)$), with the convention that $C_{i1} = 0$.
While in the monitored population, individuals may be exposed for the $m$th time at $X_{im}$ (in $\set{W}(X_{im})$ then $\set{I}(X_{im}+1)$) and subsequently exit the infectious compartment via recovery or isolation at $V_{im}$ (in $\set{I}(V_{im})$ then $\set{W}(V_{im}+1)$ or $\set{R}(V_{im}+1)$), meanwhile being tested for the $k$th time at $Z_{ik}$ for zero or more values of $k$.
Alternatively, a false positive may send an individual directly from $\set{W}(Z_{ik})$ to $\set{R}(Z_{ik}+1)$.
The indices $k$ and $m$ are cumulative and do not reset after each clearance time.
We use the convention that $Z_{i1} = V_{i1} = 0$ and $X_{i1} = -\infty$.
We write $\tilde{X}_{il}$ to mean the time of first infectious exposure in the interval between $C_{il}$ and subsequent removal from the monitored population, with the convention that $\tilde{X}_{il} = \infty$ if no such exposure occurs (i.e., due to a false positive test result), and similarly for $\tilde{V}_{il}$.
We write $L_i(t) = \max\{l : C_{il} < t\}$, $K_i(t) = \max\{k : Z_{ik} < t\}$, and $M_i(t) = \max\{m : X_{im} < t\}$ so that, e.g., $Z_{i, K_i(t)+1}$ is the time of the earliest test at or after time $t$.
Coupled with the test result indicator $Y_i(t)$, the dynamical process is then
\begin{align}
  R_i(t) &= 1
           \iff
           Y(Z_{iK_i(t)}) = 1
           \wedge
           C_{iL_i(t)} < Z_{iK_i(t)},
  \\
  W_i(t) &= 1
           \iff
           R_i(t) = 0
           \wedge
           V_{iM_i(t)} < t,
  \\
  I_i(t) &= 1
           \iff
           R_i(t) = 0
           \wedge
           V_{iM_i(t)} \geq t,
\end{align}
and $D_i(t) = 1$ if and only if $Z_{ik} = t$ for some $k$.

\section{Bias of the test-positive rate}
\label{sec:bias}

The TPR is biased as an estimator of prevalence under some natural and otherwise attractive longitudinal testing regimens.
The magnitude of this bias does not necessarily decrease with increased sample size when the population size increases proportionally.
We consider a specific form of bias that arises solely due to the longitudinal structure of the testing regimen coupled with isolation of those who test positive, even under extremely restrictive assumptions including perfect test sensitivity and specificity and independent and identically distributed testing and exposure processes between individuals.
We make the following one assumption that will be carried through the remainder of the paper, and the following two simplifying assumptions imposed only to illustrate the mechanism of bias of the test-positive rate, and which will be relaxed in later sections when considering our proposed unbiased estimator.

\begin{asm}[No undetected recoveries]
  \label{asm:no-undetected-recoveries}
  Individuals in $\set{I}$ cannot return to $\set{W}$ except by passing through $\set{R}$.
\end{asm}

Assumption~\ref{asm:no-undetected-recoveries} reflects the original goal of the surveillance scheme: to quickly detect infectious individuals and isolate them from the rest of the population.
This goal may be reasonably met (or approximately so) if the test sensitivity is high and the time between tests is sufficiently short relative to the infectious period.
The practical consequences of violating this assumption are evaluated in Section~\ref{sec:simulation-violate}.

\begin{tcnd}[Perfect test sensitivity and specificity]
  \label{tcnd:perfect-test}
  \begin{align}
    \P[Y_i(t) = 1 \giv D_i(t) = 1, W_i(t) = 0, I_i(t) = 1, R_i(t) = 0] &= 1, \\
    \P[Y_i(t) = 1 \giv D_i(t) = 1, W_i(t) = 1, I_i(t) = 0, R_i(t) = 0] &= 0.
  \end{align}
\end{tcnd}

\begin{tcnd}[Independent and identically distributed joint processes between individuals]
  \label{tcnd:iid-between-individuals}
  For all $i \neq j$,
  \begin{equation}
    \begin{aligned}
      (\vec{X}_i, \vec{V}_i, \vec{Z}_i, \vec{C}_i) &\ind (\vec{X}_j, \vec{V}_j, \vec{Z}_j, \vec{C}_j), \\
      \P[\vec{x}_i, \vec{v}_i, \vec{z}_i, \vec{c}_i] &= \P[\vec{x}_j, \vec{v}_j, \vec{z}_j, \vec{c}_j].
    \end{aligned}
  \end{equation}
\end{tcnd}
Again, both Simplifying Assumptions~\ref{tcnd:perfect-test} and \ref{tcnd:iid-between-individuals} are applied only for the remainder of this section for illustration and will be relaxed in Section~\ref{sec:identifiability}.

\subsection{Condition for unbiasedness of the test-positive rate}

Under Simplifying Assumptions~\ref{tcnd:perfect-test} and \ref{tcnd:iid-between-individuals}, the necessary and sufficient condition for unbiasedness of the TPR as an estimator of prevalence in the non-removed population is that for all non-removed individuals at any time $t$, their infectiousness and whether they are tested at that time are independent:

\begin{asm}[Marginal independence of testing and infectiousness (MITI)]
  \label{asm:miti}
  For any time $t$, and for all $i$, $D_i(t) \ind I_i(t) \giv R_i(t) = 0$.
\end{asm}

\begin{lem}
  \label{lem:unbiased-tpr}
  Under perfect test sensitivity and specificity, and independent and identically distributed joint exposure-testing processes between individuals,
  \begin{equation}
    \label{eq:marginal-unbiasedness}
    \E\left[\frac{Y_+(t)}{D_+(t)} \mgiv D_+(t) > 0\right] = \P\left[I(t) = 1 \mgiv R(t) = 0 \right]
  \end{equation}
  if and only if, for all $i$, $D_i(t) \ind I_i(t) \giv R_i(t)$ (MITI).
\end{lem}

\begin{proof}
  With perfect tests and independent and identically distributed processes,
  \begin{equation}
    \begin{aligned}
      \E\left[\frac{Y_+(t)}{D_+(t)} \mgiv D_+(t) > 0\right]
      &= \E\left[\frac{1}{D_+(t)} \sum_i D_i(t) \P\left[I_i(t) = 1 \mgiv \vec{D}(t) \right]\mgiv D_+(t) > 0\right] \\
      &= \P\left[I(t) = 1 \mgiv D(t) = 1, R(t) = 0 \right]. \\
    \end{aligned}
  \end{equation}
  Finally, we have that $\P\left[I(t) = 1 \mgiv D(t) = 1, R(t) = 0 \right] = \P\left[I(t) = 1 \mgiv R(t) = 0 \right]$ if and only if $D_i(t) \perp I_i(t) \giv R_i(t) = 0$ for all $i$ (MITI).
\end{proof}

We refer to \eqref{eq:marginal-unbiasedness} as \emph{marginal unbiasedness} of the TPR as an estimator for prevalence in the non-removed population, in contrast to the notion of \emph{conditional unbiasedness} in which the expectation of the TPR conditional on a realization of $\vec{I}(t)$ is $I_+(t) / [N-R_+(t)]$.
The sufficiency of MITI for conditional unbiasedness of the TPR follows from the usual independent sampling arguments, and MITI is necessary for conditional unbiasedness under all realizations of $\vec{I}(t)$ simultaneously because the latter is sufficient for marginal unbiasedness.

\subsection{Examples in which the test-positive rate is biased}
\label{sec:examples}

Marginal independence of testing and infectiousness (MITI) is straightforward to state and its relationship to the unbiasedness of the TPR is intuitive.
Additionally, it is easy to imagine mechanisms by which MITI would be violated: for example, individuals experiencing symptoms of the surveilled illness may volunteer to be tested earlier than they otherwise would be.
However, it has apparently gone unrecognized that MITI may be violated solely by the longitudinal nature of a testing regimen with isolation in conjunction with a strictly positive hazard of exposure, without the need for any other confounding or mediating factors.
We give four examples of testing schemes: one satisfying MITI, two violating MITI due to longitudinal scheduling and isolation alone, and one violating MITI due to confounding by another factor.
For simplicity, we assume in all examples that individuals never return to the surveilled population after being removed.

\begin{exm}[Simple random testing]
  A surveillance program tests a simple random sample of the non-removed population at each timepoint.
  MITI trivially applies and by Lemma~\ref{lem:unbiased-tpr} the TPR is an unbiased estimator of prevalence in the non-removed population.
\end{exm}

\begin{exm}[Max-gap testing]
  A surveillance program requires that no individual spend an interval greater than some maximum length in the non-removed population without being tested.
  For example, individuals may go no more than six consecutive days in the non-removed population without being tested, but may be tested more often.

  Max-gap testing does not generally satisfy MITI.
  Let $\delta$ be the maximum consecutive days an individual may spend in the non-removed population without testing.
  Then
  \begin{equation}
    \label{eq:max-gap-test-prob}
    \begin{aligned}
      \P[D(t) = 1 \giv Z_{K(t)} < t - \delta, R(t) = 0] &= 1, \\
      \P[D(t) = 1 \giv Z_{K(t)} \geq t - \delta, R(t) = 0] &\leq 1,
    \end{aligned}
  \end{equation}
  and the inequality is strict if $\P[D(t) = 1 \giv R(t) = 0] < 1$.
  Note that $\P[I(t-\delta) = 1 \giv Z_{k(t)} \geq t-\delta, R(t) = 0] = 0$, but if all non-removed individuals have a strictly positive hazard of infectious exposure each day, $\P[I(t-\delta) = 1 \giv Z_{k(t)} < t - \delta, R(t) = 0] > 0$.
  Thus if further the hazard of exposure does not depend on the time of the last test,
  \begin{equation}
    \label{eq:max-gap-exposure-prob}
    \P[I(t) = 1 \giv Z_{K(t)} < t - \delta, R(t) = 0] > \P[I(t) = 1 \giv Z_{K(t)} \geq t-\delta, R(t) = 0].
  \end{equation}
  Together, \eqref{eq:max-gap-test-prob} and \eqref{eq:max-gap-exposure-prob} violate MITI.
  
  As a concrete example, consider a scheme in which an individual last testing negative on day $t-k$, $k = 1, \ldots, 7$, has probability $k/7$ of being tested next on day $t$.
  On any given day and in the presence of no additional confounding, an individual recently tested (and therefore known to be recently non-infectious) is less likely to be tested than an individual who has been tested longer ago (and who therefore has had more opportunity to become infectious since then).
  The TPR each day would be skewed toward the prevalence among those tested longer ago and thus biased upward.
  Other distributions of waiting times could change the magnitude or direction of the bias.
\end{exm}

\begin{exm}[Once-per-period testing]
  \label{exm:once-per-period}
  A surveillance program divides the calendar into intervals $(0, b_1], (b_1, b_2], \ldots$, and within each interval $(b_{k-1}, b_{k}]$ each non-removed individual is tested exactly once.
  Suppose the intervals are weeks beginning on Monday ($b_{k-1}+1$) and ending on Sunday ($b_k$).
  There may be marginal imbalances, such as overall preference among units to be tested on certain days of the week, and within-unit correlations, such as a preference to be tested on the same day of each week.

  Once-per-period testing does not generally satisfy MITI.
  Suppose $t \in (b_{k-1}, b_k]$ and $R(t) = 0$.
  Then
  \begin{equation}
    \label{eq:once-per-period-test-prob}
    \begin{aligned}
      \P\left[D(t) = 1 \mgiv Z_{K(t)} > b_{k-1}, R(t) = 0\right] &= 0, \\
      \P\left[D(t) = 1 \mgiv Z_{K(t)} \leq b_{k-1}, R(t) = 0\right] &\geq 0,
    \end{aligned}
  \end{equation}
  and the inequality is strict if $\P[D(t) = 1] > 0$.
  Note that $\P[I(b_{k-1}+1) = 1 \giv Z_{k(t)} > b_{k-1}, R(t) = 0] = 0$, but if all non-removed individuals have a strictly positive hazard of infectious exposure each day, $\P[I(b_{k-1}+1) = 1 \giv Z_{k(t)} \leq b_{k-1}, R(t) = 0] > 0$.
  Thus if further the hazard of exposure does not depend on the time of the last test,
  \begin{equation}
    \label{eq:once-per-period-exposure-prob}
    \P\left[I(t) = 1 \mgiv Z_{K(t)} \leq b_{k-1}, R(t) = 0\right]
    >  \P\left[I(t) = 1 \mgiv Z_{K(t)} > b_{k-1}, R(t) = 0\right].
  \end{equation}
  Together, \eqref{eq:once-per-period-test-prob} and \eqref{eq:once-per-period-exposure-prob} violate MITI.

  As a concrete example, consider the specific case in which the periods are calendar weeks and each day we test a simple random sample of the non-removed population that has not yet been tested during the week (necessarily a census on the final day).
  The TPR is an unbiased estimate of the prevalence among the population eligible to be tested on that day.
  However, in the absence of additional confounding, the prevalence among the eligible population is expected to be higher than among the ineligible population because those in the ineligible population have tested negative more recently.
  Thus the TPR would be expected to overestimate the prevalence in the combined population.
  Other methods of sampling from the eligible population could change the magnitude or direction of bias.
\end{exm}

\begin{exm}[Simple random testing plus contact tracing]
  A surveillance program tests a simple random sample of the non-removed population at each timepoint, and all close contacts of those with positive test results from the previous timepoint.
  It is assumed that those tested via contact tracing are more likely to be infectious than those tested through simple random sampling.
  Thus, the overall TPR is an overestimate of prevalence, but the TPR from tests from the simple random testing component is an unbiased estimate of prevalence.
  The discrepancy between the TPRs from the two samples may provide information on transmission of the disease.
  Note that this example also violates the simplifying assumption of independence between individuals.
\end{exm}

\subsection{Magnitude of bias of the test-positive rate}

To restrict our attention to bias arising solely due to the longitudinal nature of a testing regimen in conjunction with a strictly positive hazard of exposure, without the need for any other confounding or mediating factors, we consider joint exposure-testing processes satisfying the following conditional independence assumption.

\begin{asm}[Conditional independence of testing and exposure (CITE)]
  \label{asm:cite}
  For any time $t$, and for all individuals $i$, $Z_{i, K_i(t)+1} \ind \tilde{X}_{i, L_i(t)} \giv Z_{i, K_i(t)}, C_{i, L_i(t)}, R_i(t) = 0$.
\end{asm}

Due to the assumption of no undetected recoveries, $\tilde{X}_{i,L_i(t)}$ encompasses all exposures.
CITE will be used in later sections, but for the arguments in this section all that is needed is the following relaxation.

\begin{asm}[Conditional independence of testing and infectiousness (CITI)]
  \label{asm:citi}
  For any time $t$, and for all individuals $i$, $D_i(t) \ind I_i(t) \giv Z_{i,K_i(t)}, C_{i,L_i(t)}, R_i(t) = 0$.
\end{asm}

CITI is a modification of MITI to apply within strata defined by the most recent test and clearance times, both of which are observed.
Thus in principle the TPR is an unbiased estimate of prevalence within each $(Z_{K(t)}, C_{L(t)})$ stratum \emph{for which there is a positive probability of testing}.
Simple random testing trivially satisfies CITI, and the concrete examples of max-gap and once-per-period testing given in the previous section also satisfy CITI.
However, CITI may be violated by specific examples of max-gap and once-per-period testing, for instance if symptomatic individuals tend to be tested earlier when eligible.
CITI but not CITE may be satisfied if past tests influence future behavior, e.g., if the regimen is a simple random testing scheme paired with an exposure process in which individuals tested during business hours behave more riskily that evening.

The expected prevalence at time $t$ may be decomposed as
\begin{equation}
  \P[I(t) = 1 \giv R(t) = 0]
  = \sum_{z=1}^{t-1} \P[I(t) = 1 \giv Z_{K(t)} = z, R(t) = 0] \P[Z_{K(t)} = z \giv R(t) = 0].
\end{equation}
Similarly, under CITI, the expectation of the test-positive rate at time $t$, marginally over $\vec{I}(t)$, may be decomposed as
\begin{equation}
  \begin{aligned}
    \P&[I(t) = 1 \giv D(t) = 1, R(t) = 0] \\
    &= \sum_{z=1}^{t-1} \P[I(t) = 1 \giv D(t) = 1, Z_{K(t)} = z, R(t) = 0]
    \P[Z_{K(t)} = z \giv D(t) = 1, R(t) = 0], \\
    &= \sum_{z=1}^{t-1} \P[I(t) = 1 \giv Z_{K(t)} = z, R(t) = 0]
    \P[Z_{K(t)} = z \giv D(t) = 1, R(t) = 0].
  \end{aligned}
\end{equation}
Consider the ratio of the expectation of TPR to the expectation of prevalence
\begin{equation}
  B(t)
  = \frac{
    \sum_{z=1}^{t-1} \P[I(t) = 1 \giv Z_{K(t)} = z, R(t) = 0] \P[Z_{K(t)} = z \giv D(t) = 1, R(t) = 0]
  }{
    \sum_{z=1}^{t-1} \P[I(t) = 1 \giv Z_{K(t)} = z, R(t) = 0] \P[Z_{K(t)} = z \giv \phantom{D(t) = 1,\,} R(t) = 0]
  },
\end{equation}
which can be viewed as the ratio of weighted averages of prevalences in strata defined by time of last test.
Note that scaling prevelence does not affect $B(t)$ as long as prevalence is scaled uniformly across the strata defined by time of last test.

The bias can be quite large, even for non-pathological hazards and testing schemes.
Consider a testing scheme in which every individual is tested on an $\tau$-day rotation, i.e., every time an individual is tested, their next test is $\tau$ days later (as considered in \cite{chang2021repeat}), and the same number of individuals is tested each day.
Under such a scheme,
\begin{equation}
  B(t) = \frac{
    \P[I(t) = 1 \giv Z_{K(t)} = t - \tau, R(t) = 0]
  }{
    \frac{1}{\tau} \sum_{z=t-\tau}^{t-1} \P[I(t) = 1 \giv Z_{K(t)} = z, R(t) = 0]
  }.
\end{equation}
Suppose that the infection hazard is such that approximately the same proportion $p$ of non-removed individuals are infected each day independently of time since last test, i.e., $\P[I(t) = 1 \giv Z_{K(t)} = z, R(t) = 0] \approx (t-z) p$.
An approximately constant infection probability is reasonable when the hazard is constant and small.
Then
\begin{equation}
  \begin{aligned}
    B(t)
    \approx \frac{
      \tau [t - (t - \tau)] p
    }{
      \sum_{z=t-\tau}^{t-1} (t - z) p
    }
    \approx \frac{
      \tau^2 p
    }{
      \frac{1}{2} \tau^2 p
    }
    = 2.
  \end{aligned}
\end{equation}
That is, on average the TPR over-estimates the prevalence by approximately 100\%.

\section{Identifiability of prevalence}
\label{sec:identifiability}

In the previous section we illustrated how bias of the TPR as an estimator of prevalence can arise from the longitudinal nature of the testing regimen, even under strict assumptions of perfect test sensitivity and specificity, independence and identical distributions of processes between individuals, and conditional independence of testing and exposure (CITE).
In this section we relax the first two assumptions and illustrate a Horvitz-Thompson--type (HT) estimator that is unbiased if the testing regimen is correctly specified (in the form of known testing probabilities), and nearly unbiased when the testing regimen is nonparametrically estimated in the sense that bias only arises due to Jensen's inequality applied to estimated testing probabilities.
Table~\ref{tab:assumptions} summarizes the assumptions used.

\begin{table}
  \centering
  \begin{tabular}{rp{32em}}
    \toprule
    No. & Name and abbreviated description \\
    \midrule
    \ref{asm:no-undetected-recoveries} &
                                         \textbf{No undetected recoveries:}
                                         \newline
                                         Individuals in $\set{I}$ cannot return to $\set{W}$ except by passing through $\set{R}$. \\
    \ref{asm:cite} &
                     \textbf{Conditional independence of testing and exposure (CITE):}
                     \newline
                     $Z_{i, K_i(t)+1} \perp \tilde{X}_{i, L_i(t)} \giv Z_{i, K_i(t)}, C_{i, L_i(t)}, R_i(t) = 0$. \\
    \ref{asm:simple-sens-spec} &
                                 \textbf{Simple random sensitivity and specificity:}
                                 \newline
                                 $Y_i(t) = D_i(t) \{F_i(t)(1-W_i(t)) + (1-G_i(t))W_i(t) \}$
                                 \newline
                                 with $F_i(t) \overset{iid}{\sim} \mathrm{Bernoulli}[\eta]$ and $G_i(t) \overset{iid}{\sim} \mathrm{Bernoulli}[\nu]$. \\
    \ref{asm:identically-distributed-individuals} &
                                                    \textbf{Identically distributed individuals:}
                                                    \newline
                                                    $\P[\vec{x}_i, \vec{v}_i, \vec{z}_i, \vec{c}_i] = \P[\vec{x}_j, \vec{v}_j, \vec{z}_j, \vec{c}_j]$. \\
    \ref{asm:itos} &
                     \textbf{Independence of testing from others' states (ITOS):}
                     \newline
                     $\P\left[D_i(t) = 1 \mgiv \vec{W}(t), \vec{C}_{\vec{L}(t)}\right] = \P\left[D_i(t) = 1 \mgiv W_i(t), C_{i,L_i(t)}\right]$. \\
                     \ref{asm:positivity} &
                                            \textbf{Positivity:}
                                            \newline
                                            $\P[D_i(t) = 1 \giv W_i(t) = 1, C_{i,L_i(t)} = c] > 0$. \\
    \bottomrule
  \end{tabular}
  \caption{Sufficient set of assumptions for computing an unbiased Horvitz-Thompson--type estimator of prevalence within the framework of Section~\ref{sec:processes}.}
  \label{tab:assumptions}
\end{table}

We begin by relaxing Simplifying Assumption \ref{tcnd:perfect-test} (perfect test sensitivity and specificity) and Simplifying Assumption \ref{tcnd:iid-between-individuals} (i.i.d.\ individuals) to the following assumptions, respectively.

\begin{asm}[Simple random sensitivity and specificity]
  \label{asm:simple-sens-spec}
  Positive test results are indicated by $Y_i(t) = D_i(t) \{F_i(t)(1-W_i(t)) + (1-G_i(t))W_i(t) \}$ with $F_i(t)$, $G_i(t)$ Bernoulli random variables with success probabilities $\eta \in (0, 1]$ (test sensitivity) and $\nu \in (0, 1]$ (test specificity), respectively, and independent of all other variables.
\end{asm}
\begin{asm}[Identically distributed individuals]
  \label{asm:identically-distributed-individuals}
  For all $i,j$, $\P[\vec{x}_i, \vec{v}_i, \vec{z}_i, \vec{c}_i] = \P[\vec{x}_j, \vec{v}_j, \vec{z}_j, \vec{c}_j].$
\end{asm}

\begin{asm}[Independence of testing from others' states (ITOS)]
  \label{asm:itos}
  \begin{equation*}
    \P\left[D_i(t) = 1 \mgiv \vec{W}(t), \vec{C}_{\vec{L}(t)}\right] = \P\left[D_i(t) = 1 \mgiv W_i(t), C_{i,L_i(t)}\right].
  \end{equation*}
\end{asm}

ITOS allows a priori dependence in testing schemes (e.g., preferring to test or avoiding testing members of the same household on the same day), but not dependence induced by whether or not others are in $\set{W}(t)$.
For example, ITOS would not be expected to hold if the testing program included contact tracing because an individual being removed to $\set{R}$ after recently testing positive would increase the likelihood of their close contacts being tested versus what it would have been had they tested negative and remained in $\set{W}$.

Rather than estimating the prevalence $I_+(t) / [N-R_+(t)]$ directly, we will estimate $W_+(t) = \sum_i W_i(t)$ and use the known values of $N$ and $R_+(t)$ to transform our estimate into one of prevalence.
We provide an estimator that is unbiased for $W_+(t)$ conditional on a priori unobserved $\vec{W}(t)$ via a modification of the argument of \cite{horvitz1952generalization}, weighting transformed test results by $\omega_i(t) = 1 / \P[D_i(t) = 1 \giv W_i(t) = 1, C_{i,l_i(t)} = c]$.
We will estimate $W_+^{(c)}(t) = \sum_i W_i(t) \Ind(C_{i,L_i(t)} = c)$ separately within strata defined by the observed $C_{i,L_i(t)}$ and combine the estimates into one estimate for the overall population.
The following assumption of positivity conditional on last clearance and current membership in the Well compartment (but not on intervening test times) guarantees finite weights:
\begin{asm}[Positivity]
  \label{asm:positivity}
  For any $t$ and $c < t$, $\P[D_i(t) = 1 \giv W_i(t) = 1, C_{i,l_i(t)} = c] > 0$.
\end{asm}

Theorem~\ref{thm:unbiased-estimator} provides an unbiased estimator of $W_+^{(c)}(t)$ given known inverse testing probability weights, and the subsequent Theorem~\ref{thm:test-probs} provides an expression for the testing probabilities computable under a known testing regimen or estimable from testing data.

\begin{thm}[Unbiased estimator of prevalence]
  \label{thm:unbiased-estimator}
  Assume simple random sensitivity $\eta$ and specificity $\nu$ (Assumption~\ref{asm:simple-sens-spec}), independence of testing from others' states (ITOS, Assumption~\ref{asm:itos}), and positivity (Assumption~\ref{asm:positivity}).
  Let $\omega_i(t) = 1 / \P[D_i(t) = 1 \giv W_i(t) = 1, C_{i,l_i(t)} = c]$.
  \begin{equation}
    \label{eq:estimator-expectation}
    \begin{aligned}
      \E\left[
        \frac{1}{\eta + \nu - 1} \sum_i \omega_i(t) D_i(t) \{1 - Y_i(t)\} - \frac{1-\eta}{\eta + \nu - 1} \sum_i \omega_i(t) D_i(t)
        \mgiv \vec{W}(t), \vec{C}_{\vec{L}(t)} = \vec{c}
      \right]
      = W_+(t).
    \end{aligned}
  \end{equation}
\end{thm}
Full proof of Theorem~\ref{thm:unbiased-estimator} is given in the appendix.
As a brief sketch, simple random sensitivity and specificity allows $1 - Y_i(t)$ to be replaced by $1 - \eta + W_i(t) (\eta + \nu - 1)$, ITOS allows for each $i$ the condition $\left[\vec{W}(t), \vec{C}_{\vec{L}(t)} = c\right]$ to be replaced by $\left[W_i(t), C_{L_i(t)} = c\right]$, and positivity ensures that $\omega_i(t) \P[D_i(t) = 1 \giv W_i(t), C_{i,L_i(t)} = c_i] W_i(t) = W_i(t)$ for all $i$.

\begin{thm}[Identifiability of testing probabilities]
  \label{thm:test-probs}
  Assume no undetected recoveries (Assumption~\ref{asm:no-undetected-recoveries}), conditional independence of testing and exposure (CITE, Assumption~\ref{asm:cite}), simple random sensitivity and specificity (Assumption~\ref{asm:simple-sens-spec}), and identically distributed individuals (Assumption~\ref{asm:identically-distributed-individuals}).
  Let $\Ind(z > t) = 1$ if $z > t$ and $0$ otherwise, and $\mat{P}^{(c)} = \left(p_{sz}^{(c)}\right)$ be a $(t+2)\times(t+2)$ matrix with
  \begin{equation}
    \label{eq:test-probs}
    p_{s+1, z+1}^{(c)} = \left\{
      \begin{array}{ll}
        \P\left[\min\left\{Z_{K(s+1)+1}, t+1\right\} = z \mgiv C_{L(s+1)} = c \right], & s = c, \\
        \P\left[\min\left\{Z_{K(s+1)+1}, t+1\right\} = z \mgiv D(s) = 1, Y(s) = 0, C_{L(s+1)} = c \right], & c < s \leq t, \\
        \Ind(z > t), & \textrm{otherwise}, \\
      \end{array}
    \right.
  \end{equation}
  where the $+2$ in each dimension allows the first row and first column to represent time $0$, the last column to represent time after $t$, and the last row to keep the matrix square.
  Then
  \begin{equation}
    \label{eq:test-probs-ratio}
    \begin{aligned}
      \P&[D_i(t) = 1 \giv W_i(t) = 1, C_{i,l_i(t)} = c] \\
      &= \left[\sum_{k=1}^{t-c} \nu^{k-1} \left(\mat{P}^{(c)}\right)^k \right]_{c+1,t+1}
      \Bigg/
      \sum_{z=t+1}^{t+2} \left[\sum_{k=1}^{t-c} \nu^{k-1} \left\{ \left(\mat{P}^{(c)}\right)^{k} - \left(\mat{P}^{(c)}\right)^{k-1} \right\}\right]_{c+1,z}
    \end{aligned}
  \end{equation}
  with $\mat{P}^{(c)}$ identifiable by plugging in observed proportions to its definition.
\end{thm}

Full proof of Theorem~\ref{thm:unbiased-estimator} is given in the appendix.
Key points are that no undetected recoveries and CITE imply that $\P[D_i(t) = 1 \giv W_i(t) = 1, C_{i,l_i(t)} = c]$ under the real data generating mechanism is equal to that under one in which the infection hazard is zero, and identically distributed individuals allow all individuals to share the same $\mat{P}^{(c)}$, which is a stochastic upper-triangular matrix interpretable as describing the transition probabilities among $Z_k$ as $k$ increases under zero hazard and perfect specificity.

Note that under the assumptions of Theorem~\ref{thm:test-probs} all elements of $\mat{P}^{(c)}$ admit unbiased estimators whenever at least one individual satisfies the corresponding condition, e.g.,
\begin{equation}
  \label{eq:testing-matrix-estimators}
  \begin{aligned}
    \E&\left[\frac{\sum_i \Ind\left(\min\left\{ Z_{i,K_i(c+1)+1}, t+1 \right\} = z\right) \Ind\left(C_{i,L_i(c+1)} = c\right)}{\sum_i \Ind\left(C_{i,L_i(c+1)} = c\right)} \mgiv \sum_i \Ind\left(C_{i,L_i(c+1)} = c\right) > 0\right] \\
    &= \P\left[\min\left\{Z_{K(c+1)+1}, t+1\right\} = z \mgiv C_{L(c+1)} = c \right].
  \end{aligned}
\end{equation}
When a condition is not satisfied by at least one individual, the estimator may be replaced by $\Ind(z > t)$.
Thus if the elements of $\mat{P}^{(c)}$ are known (e.g., controlled entirely by a central scheduler), prevalence may be estimated without bias, or arbitrarily small bias via Monte Carlo estimation of \eqref{eq:testing-matrix-estimators} after forward simulation of the testing process with zero exposure hazard.
If the testing probabilities are not known, they may be nonparametrically estimated from testing data via \eqref{eq:testing-matrix-estimators}.

Bias in the prevalence estimate with unknown weights arises solely from Jensen's inequality applied to the inversion of testing probabilities for weighting.
The bias therefore approaches zero asymptotically (population size going to infinity while tested proportion is held constant) due to the law of large numbers.
In practice, the bias due to noise in estimating testing probabilities appears largest when there exist non-empty strata in which the expected number of tested individuals is low, especially when there is a moderate-to-large probability of no individuals within a stratum being tested.
When no individuals are tested within a stratum, the within-stratum estimator takes the same value as if all stratum members had tested positive, yielding a high prevalence estimate that cannot be completely counterweighted by instances in which some tests are performed.
Because very small strata are most likely when incidence is low (as fewer individuals are detected and therefore cleared at the same time) and this issue has the largest effect in the same cases (because the estimator behaves as if all individuals were infectious), we propose in such situations instead estimating $W_+^{(c)}$ by the total number of non-removed individuals with $C_{i,L_i(t)} = c$, i.e., assuming all non-removed individuals in the stratum are well.
It is likely possible to evaluate the reasonableness of such a strategy in practice because prevalence estimates will be available from other strata and timepoints.

Following the argument of \cite{horvitz1952generalization}, we have under known testing probabilities and independent testing between individuals the variance expression
\begin{equation}
  \label{eq:variance}
  \begin{aligned}
    \Var&\left[
      \frac{1}{\eta + \nu - 1} \sum_i \omega_i(t) D_i(t) \{1 - Y_i(t)\} - \frac{1-\eta}{\eta + \nu - 1} \sum_i \omega_i(t) D_i(t)
      \mgiv \vec{W}(t), \vec{C}_{\vec{L}(t)} = \vec{c}
    \right] \\
    &= \left(\frac{1}{\eta + \nu - 1}\right)^2 \Var\left[
      \sum_i \omega_i(t) D_i(t) \{\eta - Y_i(t)\}
      \mgiv \vec{W}(t), \vec{C}_{\vec{L}(t)} = \vec{c}
    \right], \\
    &= \left(\frac{1}{\eta + \nu - 1}\right)^2
      \sum_{i=1}^{N} \{\eta - Y_i(t)\}^2 \frac{1 - \P[D_i(t) = 1 \giv W_i(t) = 1, C_{i,L_i(t)} = c_i]}{\left(\P[D_i(t) = 1 \giv W_i(t) = 1, C_{i,L_i(t)} = c_i]\right)^2}, \\
  \end{aligned}
\end{equation}
from which an unbiased estimator of the variance (still under known testing probabilities) may be obtained by summing over the tested individuals instead of all individuals.
\cite{horvitz1952generalization} also provides an extension to scenarios in which tests are dependent between individuals.
When estimating testing probabilities, we recommend using bias-corrected and accelerated bootstrap intervals, as illustrated in Section~\ref{sec:simulation}.

\section{Simulation study}
\label{sec:simulation}

\subsection{General setup}

We performed a simulation study to evaluate the properties of three estimators under four scenarios satisfying the assumptions required by our Horvitz-Thompson-type (HT) estimators, and five scenarios violating assumptions.
The general parameters of the assumption-satisfying scenarios were as follows.
A population of $1000$ individuals with identically distributed processes was simulated for 21 days.
Individuals were grouped into $250$ exchangeable clusters of $4$ exchangeable individuals per cluster.
The hazard of initial exposure from outside of the cluster while in the non-removed population was $h(\tau) = \frac{1}{30}\left(\frac{\tau (21 - \tau)}{(21 / 2)^2} (\frac{1}{10} - \frac{1}{50}) + \frac{1}{50}\right)$, where $\tau$ is the time since day $0$ or last clearance time.
The hazard of initial exposure from within the cluster while in the non-removed population was $1/5$ times the number of infectious individuals within the cluster, independently from exposure from outside the cluster.
The hazard of subsequent exposures was $\frac{1}{2}h(\tau)$.
Each simulation was set to begin with 2\% prevalence and peaked at approximately 5\% prevalence.
Test sensitivity and specificity were set to 83.2\% and 99.2\%, respectively, corresponding to estimates from the meta-analysis of saliva-based pCR tests for SARS-CoV-2 by \cite{butler2021comparison}.
At 5\% prevalence, these values yield 84.6\% positive predictive value and 99.1\% negative predictive value.
Individuals spent 5 days in the Removed compartment before returning to the Well compartment.

On each day within each simulation, we evaluate the test-positive rate and HT estimator with estimated testing probabilities (HT-E) as estimates of prevalence.
We produce confidence intervals at the 95\% confidence level via the exact method \citep{clopper1934use} with no finite population correction (for conservativeness) for the TPR and the bias-corrected and accelerated bootstrap approach ($BC_a$, \cite{efron1987better}) for the HT-E estimator, with 399 bootstrap iterations and acceleration factor estimated via jackknife with blocks of size 10.
When testing probabilities of individuals never exposed do not depend on exposure dynamics (e.g., as they do in the presence of contact tracing), we also give the Horvitz-Thompson estimator with known testing probabilities (HT-K), with Wald confidence intervals produced on $W_+(t)$ according to the variance formula \eqref{eq:variance} and transformed to the prevalence scale.
We simulated 1000 datasets for the TPR and HT-E estimators, and 10,000 for the HT-K estimator to account for its larger variance.

Estimates from the HT estimators are not automatically restricted to $[0, 1]$, and are sometimes below zero in the simulations above.
In practice, we recommend restricting estimates to $[0, 1]$ post hoc, as the restricted estimates are never farther than the unrestricted estimates from the truth, and are sometimes closer.
In the simulations described above, the post hoc restriction caused the HT-K estimator to be biased upward but reduced the RMSE by approximately 20\%, and did not substantially affect the HT-E estimator.

\subsection{Scenarios satisfying assumptions}
\label{sec:simulation-meet}

The scenarios satisfying the assumptions of the HT estimators are based on the first three example testing regimens in Section~\ref{sec:examples}.
In the simple random testing regimen, non-removed individuals are tested independently each day with probability $1/6$.
In the max-gap regimen, the time of first test is uniformly distributed among the first 10 days, and for subsequent times $t$ at which the most recent test or clearance time is $z$, non-removed individuals are tested with probability $(t-z)^2/10^2$.
In the once-per-period regimen, each non-removed individual is tested at a uniformly-distributed time within each 7-day calendar interval, or within the remainder of a 7-day interval in which they return to the non-removed population.
Finally, the min-max testing regimen operates similarly to the max-gap regimen but tests are not allowed within 5 days of the most recent test.
The simple random testing probability and maximum gap parameters are chosen to yield similar peak prevalences as the once-per-period regimen.

\begin{figure}
  \centering
  \includegraphics[scale = 1.0]{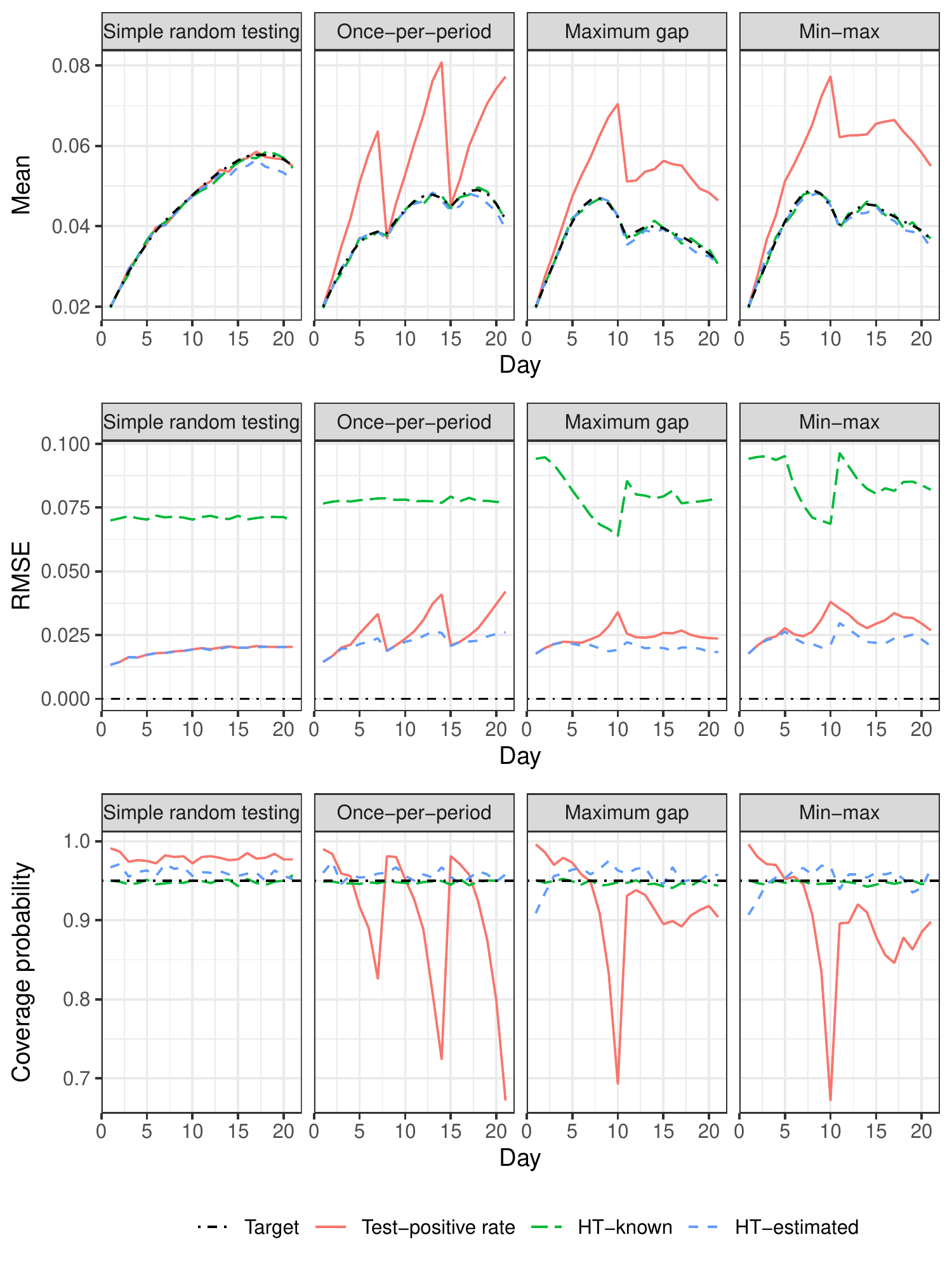}
  \caption{
    Simulations satisfying assumptions.
    Target for mean row is the mean of true prevalences across datasets.
    The RMSEs of the test-positive rate and HT-known estimator are identical under simple random testing.
    The variance of the HT-known estimator is higher than that of the HT-estimated due to occasionally very large weights.
  }
  \label{fig:sim-satisfy}
\end{figure}

Figure~\ref{fig:sim-satisfy} displays the results of the assumption-satisfying scenarios.
For the simple random testing scenario, all three estimators were approximately unbiased, and confidence interval coverage was near the nominal level, with the Clopper-Pearson intervals for the TPR being slightly conservative, as expected.
The RMSEs of the TPR and HT-E estimators are identical, and that of the HT-K substantially higher.
The substantially lower RMSE of the HT-estimated compared to HT-known estimator is likely due to a favorable bias-variance tradeoff from weight smoothing, an estimation-based relative of trimming large survey weights \citep{haziza2017construction}.
For all other assumption-satisfying scenarios, the HT estimators were unbiased and their confidence interval coverage was near the nominal level.
However, the TPR was biased upwards (except on the first day of each week in the once-per-period scenario), yielding higher RMSE and anticonservative confidence interval coverage except where the bias was small.
Coverage of TPR intervals is expected to decrease with increased sample size as the bias would remain unchanged.
In once-per-period scenario, the bias increased steadily within each period before returning to zero at the start of the next period.
In the max-gap and min-max scenarios, the first 10 days look similar to the once-per-period scenario because the first test of each individual was uniformly distributed over that period, then the bias of the test-positive rate stabilizes at roughly +30\% to +50\% as the testing hazard becomes quadratic.

\subsection{Scenarios violating assumptions}
\label{sec:simulation-violate}

The assumption-violating scenarios are based on the min-max testing regimen above because the regimen showcases both temporary ineligibility for testing and unequal testing probabilities by time of last test among those eligible.
The undetected recoveries scenario allowed an infectious individual that had not been removed within 6 days of exposure to return to the Well compartment, and exposures of those infectious at baseline were uniformly distributed among the previous 6 days.
In the time-varying sensitivity scenario (violating simple random sensitivity), a test of an infectious individual exposed at time $x$ and tested at time $t \in \{x+1, \ldots, x+10\}$ has probability $0.832 \cdot \max\left\{(t-x)(10-t+x)/(10/2)^2, 1/10\right\}$.
Pre-baseline exposures were uniformly distributed among the 6 days prior to baseline.
For estimation, we assume a sensitivity of 55.7\%, reflecting the average sensitivity during the first 10 days post-exposure (not necessarily the average sensitivity of all tests performed), though the time until detection could be longer due to imperfect sensitivity.
In the symptomatic testing scenario (violating CITE), whenever an individual is exposed, they have a $1/4$ probability of being symptomatic on their first day infectious, in which case they are tested immediately, regardless of normal eligibility rules.
In the contact tracing scenario (violating ITOS), when an individual tests positive all other non-removed individuals in their cluster are tested the next day regardless of normal eligibility rules.
In the clustered testing scenario (violating independence assumption for HT-K CIs), individuals are grouped into clusters of four, and the testing schedule is set with clusers instead of individuals as units, with individuals tested whenever the cluster is scheduled for testing and the individual is in the non-removed population.

\begin{figure}
  \centering
  \includegraphics[scale = 1.0]{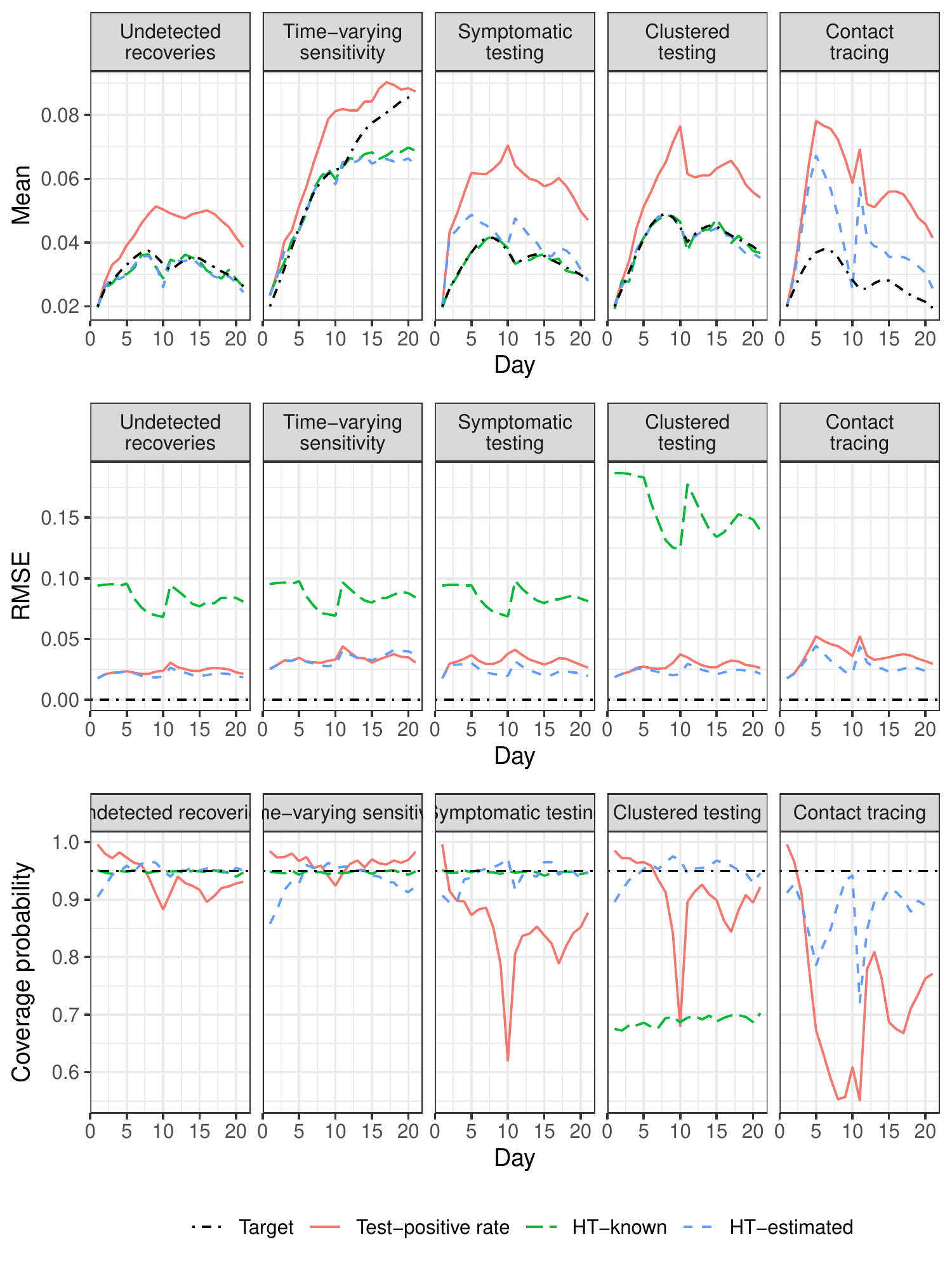}
  \caption{
    Simulations violating assumptions.
    Target for mean row is the mean of true prevalences across datasets.
    The variance of the HT-known estimator is higher than that of the HT-estimated due to occasionally very large weights.
  }
  \label{fig:sim-violate}
\end{figure}

Figure~\ref{fig:sim-violate} displays the results of the assumption-violating scenarios.
In all scenarios, the TPR was biased upward and had RMSE comparable to or higher than the HT-E estimator.
In the undetected recoveries and time-varying sensitivity scenarios, the confidence interval coverage was usually near the nominal level (though lower than in the simple random testing scenario), and in other scenarios coverage was dramatically anticonservative.
The HT-E estimator was unbiased for the clustered testing scenario, and very slightly negatively biased in the undetected recoveries scenario.
In the time-varying sensitivity scenario, the HT-E estimator was negatively biased near the end of the 21-day period, and in the asymptomatic testing and contact tracing scenarios it was substantially biased upward.
However, confidence interval coverage was near or above 90\% except in the contact tracing scenario, where it was much lower.
The HT-K estimator had similar bias to the HT-E estimator for the undetected recoveries and time-varying sensitivity scenarios, and also no bias for the clustered testing scenario, but unlike the HT-E estimator, was unbiased for the asymptomatic testing scenario (because all symptomatic individuals were infectious).
In all four of the above scenarios, the HT-K estimator had the highest RMSE but near-nominal confidence interval coverage in all but the clustered testing scenario, in which coverage was near 70\%.
The HT-K estimator is not available for the contact tracing scenario because the testing probabilities depend on the infectiousness of other cluster members.

\section{Real data analysis}
\label{sec:example}

As an illustration of our approach we analyze de-identified longitudinal testing data from 11,692 undergraduate students living on-campus at \if1\blind{[BLINDED: Large United States University]}\else{The Ohio State University}\fi\ during the fall 2020 semester.
Eligible students were required to undergo a saliva-based PCR test once per calendar workweek (Monday--Friday).
Students chose test dates subject to the once-per-week constraint.
Students could voluntarily test more than once per week and would also be tested if identified via contact tracing as potentially exposed.
There was little-to-no testing on Saturdays, Sundays, or holidays.
After taking a test, results were available within 1--2 days.
Students with positive test results were isolated for 10 days beginning the day after results were available and were exempt from further testing for 90 days following the date of their positive test.
Within the eligible population close contacts of students who received positive test results were tested and quarantined for 14 days beginning the day after results were available.
Students were sent home at the start of Thanksgiving break and the remaining class sessions were held remotely.
Figure~\ref{fig:test-counts} displays the daily test counts during the semester.

\begin{figure}
  \centering
  \includegraphics[scale = 0.8]{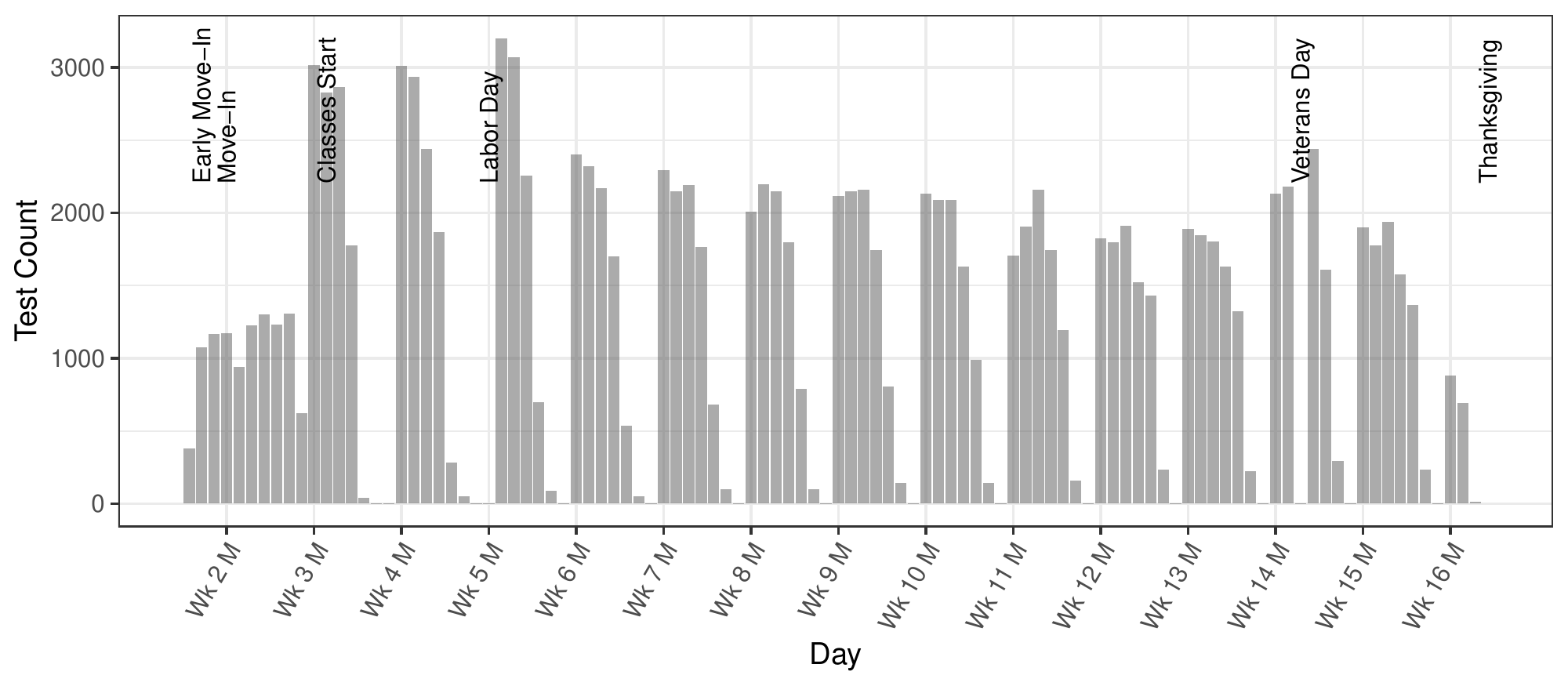}
  \caption{Daily test counts, including multiple tests per week by the same student. Vertical grid lines correspond to Mondays.}
  \label{fig:test-counts}
\end{figure}

We make the following simplifying assumptions to analyze the data.
First, we assume that all test results are returned on the second day following the test, and count those who eventually receive a positive result from a test to be part of the infectious (rather than removed) compartment through the date the result was received.
Second, when students return from isolation they are assumed to be non-infectious during the remainder of the 90 days of exemption from weekly testing.
Third, we attempted to exclude voluntary tests and tests due to contact tracing (which we expect to be non-representative) by retaining only the first test for each student during each calendar week.
Finally, although the topic of sensitivity and specificity of COVID-19 tests is complex, we assume 83.2\% test sensitivity as reported in the meta-analysis by \cite{butler2021comparison} and perfect specificity, as the vast majority of students during this semester were infection-na\"{i}ve.
Results from changes in both assumptions are also described.
Test sensitivity is assumed to be constant as a function of time since exposure.
We do not provide prevalence estimates on days for which there were less than 100 tests.
Saturdays, Sundays, and holidays comprised all but one such day.
The remaining excluded day was the Friday following the first day of class, for which the reason for the low test count is unknown to us.

To anonymize the data, we first construct a matrix with days as columns and students as rows.
Elements are the results of tests taken on that day (positive/negative) if applicable, or missing if no test was taken.
No student identifiers are present, and the order of the matrix rows is randomized.
Finally, iterating forward through the days, students in the non-removed population were stratified by last test date and last clearance date, and the vectors of subsequent test times and results were permuted within strata (retaining ordering within vectors).
This final shuffling mitigates risk of student identification via their longitudinal testing sequences but does not change the values of any of the estimators considered.

\begin{figure}
  \centering
  \includegraphics[scale = 0.8]{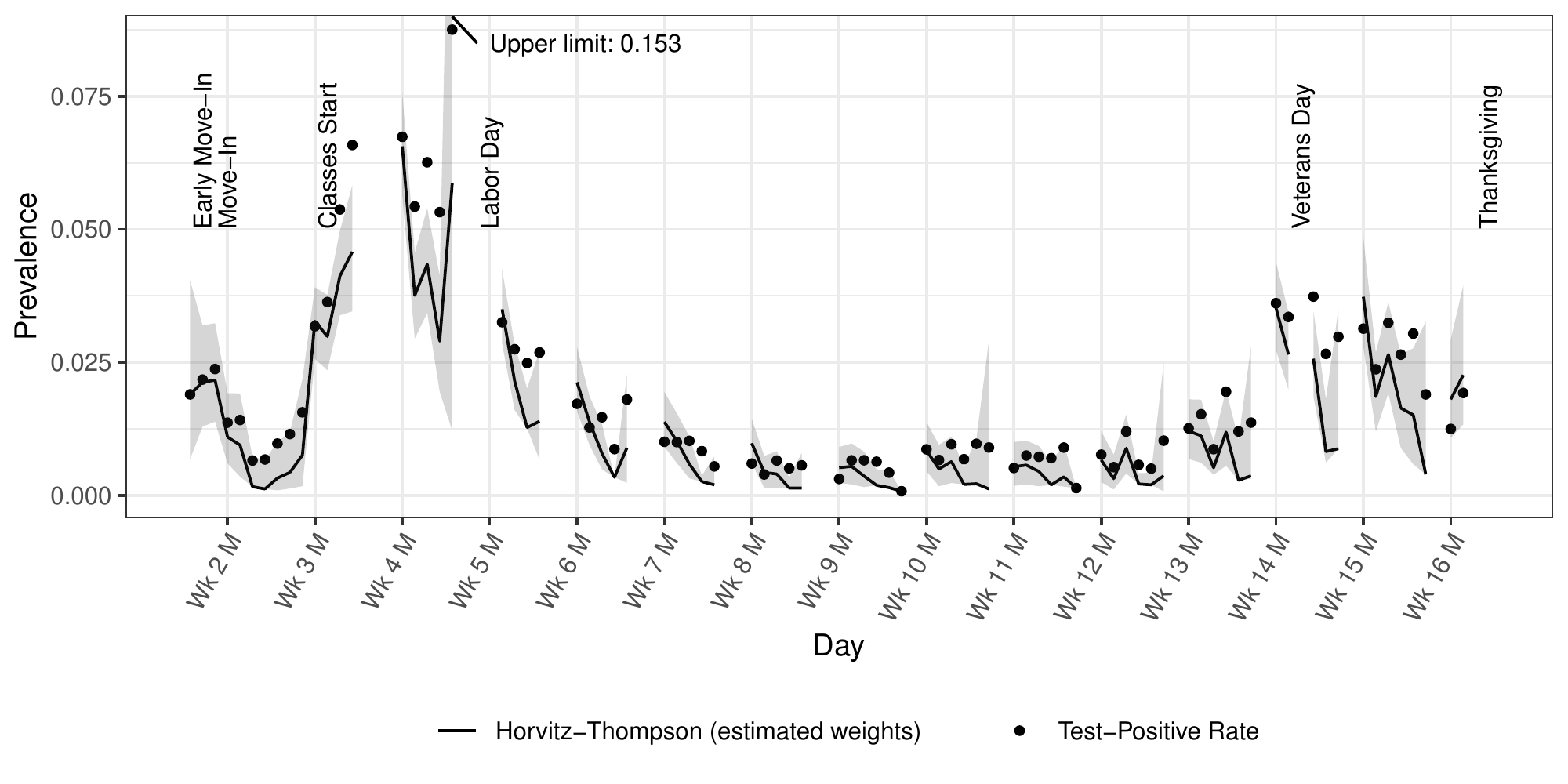}
  \caption{Daily prevalence estimates, adjusted for test sensitivity and reporting delays. Shaded ribbon is $BC_a$ 95\% confidence band. No estimates for days with $<100$ tests taken. Vertical grid lines correspond to Mondays.}
  \label{fig:result-adj}
\end{figure}

Figure~\ref{fig:result-adj} displays daily test-positive rates and HT-E prevalence estimates, both adjusted for imperfect test sensitivity, delay of test results, and assumed non-infectiousness during the time exempt from testing following the end of isolation.
The HT-E $BC_a$ confidence bounds were produced using 999 bootstrap iterations and 79 blocks of 148 individuals for the jackknife-estimated acceleration factor.
Up to the start of classes, both estimators largely agree, though the TPR is slightly higher, especially late in Week 2, consistent with the results of the simulation study under once-per-period testing.
The drop in both estimates immediately following move-in suggests a baseline prevalence of roughly 2\% among students moving in that was reduced via testing, isolation, and quarantine.

Following the start of classes, the HT-E estimate generally decreases over the course of each week, consistent with a mechanism in which a relatively large number of individuals are infected during the weekends then detected and removed by weekly testing.
The HT-E estimates but not the TPRs reflect the observation by \cite{hay2021intuition} that ``isolating positive individuals reduces prevalence in the tested population.''
Meanwhile, within each week the TPRs increase relative to the HT-E estimates, though the TPRs do not increase over the course of every week in absolute terms.

Varying test sensitivity between 77.4\% and 91.4\%, reflecting the 95\% confidence intervals reported by \cite{butler2021comparison}, yielded similar patterns and conclusions but scaled estimates within 100--127\% and 76--100\%, respectively.
To evaluate the effect of imperfect specificity, we assumed a value of 99.2\%, also following \cite{butler2021comparison}, with sensitivity ranging from 83.2\% to 50\%.
Although the results were similar on the semester scale, estimates that were below approximately 1\% in Figure~\ref{fig:result-adj} dropped to zero (TPR) or between 0 and 0.2\% (HT-E).
We do not consider even this high estimate of imperfect specificity to be realistic: a third of crude test-positive rates are less than 1\%, implying that almost all of the positive test results those days were false positives.

\section{Discussion}
\label{sec:discussion}

We have identified and characterized an under-recognized bias in prevalence estimates based on test-positive rates of repeated screening tests when those testing positive are subsequently isolated, and have presented unbiased and approximately unbiased estimators of prevalence in such situations.
The bias in question arises under natural repeated testing regimens such as once-per-week testing, and is present even when tests have perfect sensitivity and specificity, and without confounding factors such as contact tracing or symptomatic testing.
This bias arises due to confounding between the probability of an individual being tested on a given day and the probability that they are infectious, caused in part by the constraints of the testing schedule.
Our estimator achieves unbiasedness by weighting test results by the inverse probability of testing under a hypothetical scenario with zero hazard of infection, which may be estimated directly from the data.

We have illustrated the bias of test-positive rates and unbiasedness (or approximate unbiasedness) of our estimators via simulation studies under complications of imperfect test sensitivity and specificity.
We have also proposed $BC_a$ bootstrap confidence intervals which are straightforward to implement and appear well-calibrated in the correctly-specified simulation study but do not account for real data complications such as clustering of test schedules.
Further development of confidence interval constructions would be important.

Analysis of once-per-week testing data from the fall 2020 semester at \if1\blind{[BLINDED: Large United States University]}\else{The Ohio State University}\fi\ illustrated the feasibility of handling complications such as non-compliance to the testing regimen and via crude adjustments contact tracing, reporting delays, and temporary exemption from testing post-isolation.
Although on multi-week timescales the prevalence curves given by the test-positive rates and our estimator broadly agreed, the TPR tended to be higher and we identified systematic within-week discrepancies illustrating the bias of the TPR and suggesting a different weekly timing of incidence (higher on weekends rather than uniformly throughout the week) that could have implications for the efficacy of on-campus safety measures (e.g., social transmission versus in-classroom transmission).

Our proposed estimator and analysis relied on assumptions in three classes.
First, verifiable conditions on the testing process design that simplify considerations but can be influenced by the surveillant, and violations of which could be handled by book-keeping modifications of our estimator (such as differentiating between scheduled tests and those induced by contact tracing or symptoms).
Second, no-confounding assumptions that preclude alterations of risk-relevant behavior in response to scheduled tests, or of testing schedules in response to perceived risks, are essential to our theoretical results.
However, violations of these do not necessarily cause our estimators to perform worse than the TPR.
For example, if individuals believing they may have been exposed tended to schedule tests earlier than they otherwise would have, we would expect to see bias patterns similar to those of the symptomatic testing or contact tracing simulation scenarios, even though it may be more difficult to account for via book-keeping.
Finally, it may be possible to relax simplifying assumptions on the technical deteails of disease and testing processes such as known time-invariant test sensitivity \citep{chang2021repeat} and a long infectious period relative to gaps between tests, though any relaxation would likely introduce significantly more complexity and is reserved for future work.

Our strategy has been to provide an estimation approach to achieving unbiasedness while remaining as close as possible mechanically to the test-positive rate analysis.
We focus exclusively on prevalence estimation without modeling transmission or incorporating external data sources to emphasize correction of the repeated testing bias.
Combining such approaches with ours could greatly improve the accuracy of estimates, potentially at the expense of ease of implementation or robustness.
A promising alternative approach to constructing estimators under similar assumptions is via a hazard-based framework in the style of time-to-event analyses \citep{khudabukhsh2020survival}.
Under such an approach, some of the independence assumptions may be reinterpretable as independent censoring.

\bigskip
\begin{center}
{\large\bf SUPPLEMENTARY MATERIAL}
\end{center}

\begin{description}

\item[Data and code:] The supplementary materials include R scripts and shuffled data for reproducing the simulation and analysis results.
  
\item[Appendix:] The appendix contains proofs of results described in the main text.

\end{description}

\bigskip
\begin{center}
{\large\bf ACKNOWLEDGMENTS}
\end{center}

The authors thank \if1\blind{[BLINDED]}\else{Eben Kenah}\fi, \if1\blind{[BLINDED]}\else{Mikkel Quam}\fi, and \if1\blind{[BLINDED]}\else{Rebecca Andridge}\fi\, for their advice regarding epidemiological considerations, details of the testing regimen of \if1\blind{[BLINDED: Large United States University]}\else{The Ohio State University}\fi, and survey estimation methods, respectively.

\bibliographystyle{apalike}
\bibliography{1-article.bib}

\clearpage

\renewcommand{\theequation}{A.\arabic{equation}}
\renewcommand{\thesection}{A.\arabic{section}}

\setcounter{equation}{0}
\setcounter{section}{0}
\setcounter{thm}{0}

\section{Unbiased estimator of prevalence}

\begin{thm}[Unbiased estimator of prevalence]
  \label{Athm:unbiased-estimator}
  Assume simple random sensitivity $\eta$ and specificity $\nu$ (Assumption~\ref{asm:simple-sens-spec}), independence of testing from others' states (ITOS, Assumption~\ref{asm:itos}), and positivity (Assumption~\ref{asm:positivity}).
  Let $\omega_i(t) = 1 / \P[D_i(t) = 1 \giv W_i(t) = 1, C_{i,l_i(t)} = c]$.
  Then
  \begin{equation}
    \label{Aeq:estimator-expectation}
    \begin{aligned}
      \E\left[
        \frac{1}{\eta + \nu - 1} \sum_i \omega_i(t) D_i(t) \{1 - Y_i(t)\} - \frac{1-\eta}{\eta + \nu - 1} \sum_i \omega_i(t) D_i(t)
        \mgiv \vec{W}(t), \vec{C}_{\vec{L}(t)} = \vec{c}
      \right]
      = W_+(t).
    \end{aligned}
  \end{equation}
\end{thm}

\begin{proof}
  Let $\omega_i(t)$ be finite weights, not necessarily summing to $1$, associated with individual $i$ at time $t$.
  Under simple random sensitivity and specificity, plus ITOS,
  \begin{equation}
    \label{Aeq:estimator-expectation}
    \begin{aligned}
      \E&\left[
        \frac{1}{\eta + \nu - 1} \sum_i \omega_i(t) D_i(t) \{1 - Y_i(t)\} - \frac{1-\eta}{\eta + \nu - 1} \sum_i \omega_i(t) D_i(t)
        \mgiv \vec{W}(t), \vec{C}_{\vec{L}(t)} = \vec{c}
      \right] \\
      &= \frac{1}{\eta + \nu - 1} \sum_i \omega_i(t) \P[D_i(t) = 1 \giv \vec{W}(t), \vec{C}_{\vec{L}(t)} = \vec{c}] \{1 - \eta + W_i(t) (\eta + \nu - 1)\} \\
      &\quad\quad- \frac{1-\eta}{\eta + \nu - 1} \sum_i \omega_i(t) \P[D_i(t) = 1 \giv \vec{W}(t), \vec{C}_{\vec{L}(t)} = \vec{c}], \\
      &= \sum_i \omega_i(t) \P[D_i(t) = 1 \giv \vec{W}(t), \vec{C}_{\vec{L}(t)} = \vec{c}] W_i(t), \\
      &= \sum_i \omega_i(t) \P[D_i(t) = 1 \giv W_i(t), C_{i,L_i(t)} = c_i] W_i(t),
    \end{aligned}
  \end{equation}
  where $\eta \in (0, 1]$ and $\nu \in (0, 1]$ are the test sensitivity and specificity, respectively.
  Then, under positivity, choosing the finite weights
  \begin{equation}
    \label{Aeq:weights}
    \omega_i(t) = 1 / \P[D_i(t) = 1 \giv W_i(t) = 1, C_{i,L_i(t)} = c_i]
  \end{equation}
  yields $\omega_i(t) \P[D_i(t) = 1 \giv W_i(t), C_{i,L_i(t)} = c_i] W_i(t) = W_i(t)$ for all $i$, making \eqref{Aeq:estimator-expectation} equal to $W_+(t)$.
\end{proof}

\section{Identifiability of testing probabilities}

\begin{thm}[Identifiability of testing probabilities]
  \label{Athm:test-probs}
  Assume no undetected recoveries (Assumption~\ref{asm:no-undetected-recoveries}), conditional independence of testing and exposure (CITE, Assumption~\ref{asm:cite}), simple random sensitivity and specificity (Assumption~\ref{asm:simple-sens-spec}), and identically distributed individuals (Assumption~\ref{asm:identically-distributed-individuals}).
  Let $\Ind(z > t) = 1$ if $z > t$ and $0$ otherwise, and $\mat{P}^{(c)} = \left(p_{sz}^{(c)}\right)$ be a $(t+2)\times(t+2)$ matrix with
  \begin{equation}
    \label{Aeq:test-probs}
    p_{s+1, z+1}^{(c)} = \left\{
      \begin{array}{ll}
        \P\left[\min\left\{Z_{K(s+1)+1}, t+1\right\} = z \mgiv C_{L(s+1)} = c \right], & s = c, \\
        \P\left[\min\left\{Z_{K(s+1)+1}, t+1\right\} = z \mgiv D(s) = 1, Y(s) = 0, C_{L(s+1)} = c \right], & c < s \leq t, \\
        \Ind(z > t), & \textrm{otherwise}, \\
      \end{array}
    \right.
  \end{equation}
  where the $+2$ in each dimension allows the first row and first column to represent time $0$, the last column to represent time after $t$, and the last row to keep the matrix square.
  Then
  \begin{equation}
    \label{Aeq:test-probs-ratio}
    \begin{aligned}
      \P&[D_i(t) = 1 \giv W_i(t) = 1, C_{i,l_i(t)} = c] \\
      &= \left[\sum_{k=1}^{t-c} \nu^{k-1} \left(\mat{P}^{(c)}\right)^k \right]_{c+1,t+1}
      \Bigg/
      \sum_{z=t+1}^{t+2} \left[\sum_{k=1}^{t-c} \nu^{k-1} \left\{ \left(\mat{P}^{(c)}\right)^{k} - \left(\mat{P}^{(c)}\right)^{k-1} \right\}\right]_{c+1,z}
    \end{aligned}
  \end{equation}
  with $\mat{P}^{(c)}$ identifiable by plugging in observed proportions to its definition.
\end{thm}

\begin{proof}
  Under the assumption of no undetected recoveries,
  \begin{equation}
    \label{Aeq:identifiable-testing-prob}
    \begin{aligned}
      \P&[D_i(t) = 1 \giv W_i(t) = 1, C_{i,L_i(t)} = c] \\
      &= \P[D_i(t) = 1 \giv L_i(t) = L_i(c+1), W_i(t) = 1, C_{i,L_i(c_i+1)} = c], \\
      &= \P[D_i(t) = 1 \giv L_i(t) = L_i(c+1), R_i(t) = 0, \tilde{X}_{i, L_i(c_i+1)} \geq t, C_{i,L_i(c_i+1)} = c], \\
      &=\frac{
        \P[D_i(t) = 1, L_i(t) = L_i(c+1), R_i(t) = 0 \giv  \tilde{X}_{i, L_i(c+1)} \geq t, C_{i,L_i(c+1)} = c]
      }{
        \P[L_i(t) = L_i(c+1), R_i(t) = 0 \giv  \tilde{X}_{i, L_i(c+1)} \geq t, C_{i,L_i(c+1)} = c]
      },
    \end{aligned}
  \end{equation}
  because $C_{i,L_i(t)} = c$ if and only if $C_{i,L_i(c+1)} = c$ and $L_i(t) = L_i(c+1)$, and given that there have been no additional clearances after $c$, $W_i(t) = 1$ if and only if the individual remains unexposed ($\tilde{X}_{i,L_i(c+1)} \geq t$) and has not been erroneously removed ($R_i(t) = 0$).
  The denominator is $1$ under perfect specificity.

  We have the recursive decomposition for the numerator
  \begin{equation}
    \begin{aligned}
      \P&[D_i(t) = 1, L_i(t) = L_i(c+1), R_i(t) = 0 \giv  \tilde{X}_{L_i(c+1)} \geq t, C_{i,L(c+1)} = c] \\
      &= \sum_{c < s < t}\left(
        \begin{array}{l}
          \P\left[\begin{array}{l} Z_{i,K_i(s+1)+1} = t, \\ L_i(t) = L_i(c+1), R_i(t) = 0\end{array} \mgiv \begin{array}{l} D_i(s) = 1, L_i(s) = L_i(c+1), \\ \tilde{X}_{i,L_i(c+1)} \geq t, C_{i,L_i(c+1)} = c\end{array}\right] \\
          \times \P[D_i(s) = 1, L_i(s) = L_i(c+1) \giv \tilde{X}_{i,L_i(c+1)} \geq t, C_{i,L_i(c+1)} = c]
        \end{array}
      \right) \\
      &\quad\quad+ \P\left[Z_{i,K_i(c+1)+1} = t \mgiv C_{i,L_i(c+1)} = c\right], \\
      &= \sum_{c < s < t}\left(
        \begin{array}{l}
          \P\left[\begin{array}{l} Z_{i,K_i(s+1)+1} = t, \\ L_i(t) = L_i(c+1), R_i(t) = 0\end{array} \mgiv \begin{array}{l} D_i(s) = 1, L_i(s) = L_i(c+1), \\ \tilde{X}_{i,L_i(c+1)} \geq s, C_{i,L_i(c+1)} = c\end{array}\right] \\
          \times \P[D_i(s) = 1, L_i(s) = L_i(c+1) \giv \tilde{X}_{i,L_i(c+1)} \geq s, C_{i,L_i(c+1)} = c]
        \end{array}
      \right) \\
      &\quad\quad+ \P\left[Z_{i,K_i(c+1)+1} = t \mgiv C_{i,L_i(c+1)} = c\right], \\
      &= \sum_{c < s < t}\left(
        \begin{array}{l}
          \P\left[\begin{array}{l} Z_{i,K_i(s+1)+1} = t, \\ L_i(t) = L_i(c+1), R_i(t) = 0\end{array} \mgiv \begin{array}{l} D_i(s) = 1, L_i(s+1) = L_i(c+1), \\ Y_i(s) = 0, C_{i,L_i(c+1)} = c\end{array}\right] \\
          \times \nu \P[D_i(s) = 1, L_i(s) = L_i(c+1) \giv \tilde{X}_{i,L(c+1)} \geq s, C_{i,L_i(c+1)} = c]
        \end{array}
      \right) \\
      &\quad\quad+ \P\left[Z_{i,K_i(c+1)+1} = t \mgiv C_{i,L_i(c+1)} = c\right], \\
      &= \sum_{c < s < t}\left(
        \begin{array}{l}
          \P\left[Z_{i,K_i(s+1)+1} = t, R_i(t) = 0 \mgiv D_i(s) = 1, Y_i(s) = 0, C_{i,L_i(s+1)} = c \right] \\
          \times \nu \P[D_i(s) = 1, L_i(s) = L_i(c+1) \giv \tilde{X}_{i,L_i(c+1)} \geq s, C_{i,L_i(c+1)} = c]
        \end{array}
      \right) \\
      &\quad\quad+ \P\left[Z_{i,K_i(c+1)+1} = t \mgiv C_{i,L_i(c+1)} = c\right],
    \end{aligned}
  \end{equation}
  where the third equality is due to CITE and the fourth due to simple random sensitivity and specificity.
  Similarly,
  \begin{equation}
    \begin{aligned}
      \P&[D_i(t) = 0, L_i(t) = L_i(c+1), R_i(t) = 0 \giv  \tilde{X}_{i,L_i(c+1)} \geq t, C_{i,L_i(c+1)} = c] \\
      &= \sum_{c < s < t}\left(
        \begin{array}{l}
          \P\left[Z_{i,K_i(s+1)+1} > t, R_i(t) = 0 \mgiv D_i(s) = 1, Y_i(s) = 0, C_{i,L_i(s+1)} = c \right] \\
          \times \nu \P[D_i(s) = 1, L_i(s) = L_i(c+1) \giv \tilde{X}_{i,L_i(c+1)} \geq s, C_{i,L_i(c+1)} = c]
        \end{array}
      \right) \\
      &\quad\quad+ \P\left[Z_{i,K_i(c+1)+1} > t \mgiv C_{i,L_i(c+1)} = c\right],
    \end{aligned}
  \end{equation}
  so that
  \begin{equation}
    \begin{aligned}
      \P&[L_i(t) = L_i(c+1), R_i(t) = 0 \giv  \tilde{X}_{i,L_i(c+1)} \geq t, C_{i,L_i(c+1)} = c] \\
      &= \sum_{d=0}^{1} \P[D_i(t) = d, L_i(t) = L_i(c+1), R_i(t) = 0 \giv  \tilde{X}_{i,L_i(c+1)} \geq t, C_{i,L_i(c+1)} = c], \\
      &= \sum_{c < s < t}\left(
        \begin{array}{l}
          \P\left[Z_{i,K_i(s+1)+1} \geq t, R_i(t) = 0 \mgiv D_i(s) = 1, Y_i(s) = 0, C_{i,L_i(s+1)} = c \right] \\
          \times \nu \P[D_i(s) = 1, L_i(s) = L_i(c+1) \giv \tilde{X}_{i,L_i(c+1)} \geq s, C_{i,L_i(c+1)} = c]
        \end{array}
      \right) \\
      &\quad\quad+ \P\left[Z_{i,K_i(c+1)+1} \geq t \mgiv C_{i,L_i(c+1)} = c\right].
    \end{aligned}
  \end{equation}

  Let $\Ind(z > t) = 1$ if $z > t$ and $0$ otherwise, and $\mat{P}^{(i,c)} = \left(p_{sz}^{(i,c)}\right)$ be a $(t+2)\times(t+2)$ matrix with
  \begin{equation}
    p_{s+1, z+1}^{(i,c)} = \left\{
      \begin{array}{ll}
        \P\left[\min\left\{Z_{i,K_i(s+1)+1}, t+1\right\} = z \mgiv C_{i,L_i(s+1)} = c \right], & s = c, \\
        \P\left[\min\left\{Z_{i,K_i(s+1)+1}, t+1\right\} = z \mgiv D_i(s) = 1, Y_i(s) = 0, C_{i,L_i(s+1)} = c \right], & c < s \leq t, \\
        \Ind(z > t), & \textrm{otherwise}, \\
      \end{array}
    \right.
  \end{equation}
  where the $+2$ in each dimension allows the first row and first column to represent time $0$, the last column to represent time after $t$, and the last row to keep the matrix square.
  With $s = c$, $\P\left[\min\left\{Z_{i,K_i(s+1)+1}, t+1\right\} = z \mgiv C_{i,L_i(s+1)} = c \right] = \P\left[Z_{i,K_i(c+1)+1} = z \mgiv C_{i,L_i(c+1)} = c \right]$ if $z \leq t$ and $\P\left[\min\left\{Z_{i,K_i(s+1)+1}, t+1\right\} = z \mgiv C_{i,L_i(s+1)} = c \right] = \P\left[Z_{i,K_i(c+1)+1} > z \mgiv C_{i,L_i(c+1)} = c \right]$ if $z > t$, and similarly for $\P\left[\min\left\{Z_{i,K_i(s+1)+1}, t+1\right\} = z \mgiv D_i(s) = 1, Y_i(s) = 0, C_{i,L_i(s+1)} = c \right]$ when $c < s \leq t$.
  
  The matrix $\mat{P}^{(i,c)}$ is stochastic upper-triangular with zeros along the diagonal except at the bottom-right corner, and may be interpreted as describing the transition probabilities among $Z_k$ as $k$ increases when the exposure hazard is uniformly zero and the test has perfect specificity.
  The quantity $\left[\nu^{k-1} \left(\mat{P}^{(i,c)}\right)^k \right]_{c+1,t+1}$ describes the probability of an individual, with last clearance time $c$ and no exposures before $t$, having their $k$th test after clearance at time $t$ with no intervening false positives.
  Similarly, $\left[\nu^{k-1} \left(\mat{P}^{(i,c)}\right)^k \right]_{c+1,t+2}$ describes the probability of an individual, with last clearance time $c$ and no exposures before $t$, having their $k$th test after clearance after time $t$ (and zero or more prior tests at or after time $t$) with no intervening false positives.
  Thus,
  \begin{equation}
    \label{Aeq:matrix-expressions}
    \begin{aligned}
      &\left[\sum_{k=1}^{t-c} \nu^{k-1} \left(\mat{P}^{(i,c)}\right)^k \right]_{c+1,t+1} \\
      &\quad\quad= \P[D_i(t) = 1, L_i(t) = L_i(c+1), R_i(t) = 0 \giv  X_{i,L_i(c+1)} \geq t, C_{i,L_i(c+1)} = c], \\
      &\sum_{z=t+1}^{t+2} \left[\sum_{k=1}^{t-c} \nu^{k-1} \left\{ \left(\mat{P}^{(i,c)}\right)^{k} - \left(\mat{P}^{(c)}\right)^{k-1} \right\}\right]_{c+1,z} \\
      &\quad\quad= \P[L_i(t) = L_i(c+1), R_i(t) = 0 \giv  X_{i,L_i(c+1)} \geq t, C_{i,L_i(c+1)} = c].
    \end{aligned}
  \end{equation}
  A version of result \eqref{Aeq:test-probs} specific to individual $i$ then follows from substituting the expressions \eqref{Aeq:matrix-expressions} into the ratio \eqref{Aeq:test-probs-ratio}.

  Under identically distributed individuals, we may define $\mat{P}^{(c)}$ to be $\mat{P}^{(i,c)}$ with the subscript $i$ dropped everywhere in its definition.
  We may then estimate $\mat{P}^{(c)}$ due to the following equations:
  \begin{equation}
    \label{Aeq:testing-matrix-estimators}
    \begin{aligned}
      \E&\left[\frac{\sum_i \Ind\left(\min\left\{ Z_{i,K_i(c+1)+1}, t+1 \right\} = z\right) \Ind\left(C_{i,L_i(c+1)} = c\right)}{\sum_i \Ind\left(C_{i,L_i(c+1)} = c\right)} \mgiv \sum_i \Ind\left(C_{i,L_i(c+1)} = c\right) > 0\right] \\
      &= \P\left[\min\left\{Z_{K(c+1)+1}, t+1\right\} = z \mgiv C_{L(c+1)} = c \right], \\
      \E&\left[\frac{
          \sum_i \left\{
            \begin{array}{l}
              \Ind\left(\min\left\{ Z_{i,K_i(c+1)+1}, t+1 \right\} = z\right) \\
              \times D_i(s) \left\{1 - Y_i(s)\right\} \Ind\left(C_{i,L_i(c+1)} = c\right)
            \end{array}
          \right\}
        }{
          \sum_i D_i(s) \left\{1 - Y_i(s)\right\} \Ind\left(C_{i,L_i(c+1)} = c\right)} \mgiv \sum_i D_i(s) \left\{1 - Y_i(s)\right\}\Ind\left(C_{i,L_i(c+1)} = c\right) > 0\right] \\
      &= \P\left[\min\left\{Z_{K(c+1)+1}, t+1\right\} = z \mgiv D(s) = 1, Y(s) = 0, C_{L(c+1)} = c \right]. \\
    \end{aligned}
  \end{equation}
  When a condition above is not satisfied by at least one individual, we replace the estimator with $\Ind(z > t)$.
\end{proof}

\end{document}